\documentclass[11pt]{article}
\usepackage{amsmath,amssymb}
\usepackage{xfrac}
\usepackage[margin=1in]{geometry}
\usepackage{graphicx}
\graphicspath{{fig/}}
\usepackage{tikz}
\usetikzlibrary{quantikz}
\usepackage{authblk}
\usepackage{hyperref,cleveref}

\usepackage{amsthm}

\newtheorem{proposition}{Proposition}

\newtheorem*{proposition*}{Proposition}

\newcommand{\mbb}{\mathbb}
\newcommand{\abs}[1]{\left| #1 \right|}
\newcommand{\norm}[1]{\left\| #1 \right\|}
\newcommand{\set}[1]{\left\{ #1 \right\}}

\newcommand{\inv}{^{-1}}
\newcommand{\adjoint}{^{\dagger}}
\newcommand{\grad}{\nabla}

\newcommand{\tensor}{\otimes}
\newcommand{\hilbert}{\mathcal{H}}

\renewcommand{\hat}{\widehat}
\renewcommand{\tilde}{\widetilde}
\newcommand{\domain}{\mathcal{D}}
\renewcommand{\S}{Section}

\title{Improving the variational quantum eigensolver using variational adiabatic quantum computing}

\author[1]{Stuart M. Harwood%
\thanks{Corresponding author: stuart.m.harwood@exxonmobil.com}}
\author[1]{Dimitar Trenev}
\author[1]{Spencer T. Stober}
\author[2]{Panagiotis Barkoutsos}
\author[3]{Tanvi P. Gujarati}
\author[4]{Sarah Mostame}
\author[4]{Donny Greenberg}

\affil[1]{ExxonMobil Corporate Strategic Research, Annandale, NJ 08801, USA}
\affil[2]{IBM Quantum, IBM Research Zurich, 8803 R{\"u}schlikon, Switzerland}
\affil[3]{IBM Quantum, IBM Research Almaden, San Jose, CA 95120, USA}
\affil[4]{IBM Quantum, IBM T.J. Watson Research Center, Yorktown Heights, NY 10598, USA}

\date{\today}
\begin{document}
\maketitle

\begin{abstract}
The variational quantum eigensolver (VQE) is a hybrid quantum-classical algorithm for finding the minimum eigenvalue of a Hamiltonian that involves the optimization of a parameterized quantum circuit.
Since the resulting optimization problem is in general nonconvex, the method can converge to suboptimal parameter values which do not yield the minimum eigenvalue.
In this work, we address this shortcoming by adopting the concept of variational adiabatic quantum computing (VAQC) as a procedure to improve VQE.
In VAQC, the ground state of a continuously parameterized Hamiltonian is approximated via a parameterized quantum circuit.
We discuss some basic theory of VAQC to motivate the development of a hybrid quantum-classical homotopy continuation method.
The proposed method has parallels with a predictor-corrector method for numerical integration of differential equations.
While there are theoretical limitations to the procedure, we see in practice that VAQC can successfully find good initial circuit parameters to initialize VQE.
We demonstrate this with two examples from quantum chemistry.
Through these examples, we provide empirical evidence that VAQC, combined with other techniques
(an adaptive termination criteria for the classical optimizer and a variance-based resampling method for the expectation evaluation), can provide more accurate solutions than ``plain'' VQE, for the same amount of effort.
\end{abstract}

%%%%%%%%%%%%%%%%%%%%%%%%%%%%%%%%%%%%%%%%%%%%%%%%%%%%%%%%%%%%%%
%%%%%%%%%%%%%%%%%%%%%%%%%%%%%%%%%%%%%%%%%%%%%%%%%%%%%%%%%%%%%%
\section{Introduction}

The variational quantum eigensolver (VQE) \cite{mcclean2016theory,peruzzo_2014,moll2018quantum} is a hybrid quantum-classical algorithm for approximating the minimum eigenvalue of a given Hamiltonian.
A parameterized quantum circuit, or ansatz, is tuned via a classical optimization method in order to minimize the expected value of the Hamiltonian.
This expected value is estimated on a quantum computer, and thus this evaluation is noisy
(at the very least, there is noise from the finite number of circuit samples or shots used to estimate the expected value).
Consequently, VQE has characteristics of a stochastic optimization problem.
VQE has been used to estimate the ground state molecular energy of small molecules \cite{peruzzo_2014,kandala_2017}, and as a heuristic for Ising problems \cite{moll2018quantum,nannicini2019performance}.

While VQE has shown promise in some small size problems, its ability to scale to larger problems is an open question;
ultimately, the classical optimization methods on which VQE relies are at best local methods, while the objectives are typically nonconvex with multiple local minima.
Convergence to a local, but not global, minimizer limits the accuracy of the estimated eigenvalue.

To address this fundamental shortcoming of VQE, as well as the parametric nature of certain applications, we adopt the concept of variational adiabatic quantum computing (VAQC).
In this approach, we use VQE to calculate the minimum eigenvalues of a parameterized Hamiltonian, taking advantage of the continuity of the problems to bootstrap or warm-start the local optimization methods.
At a high level, we see that VAQC uses a parameterized Hamiltonian to navigate an energy landscape in order to arrive at the minimum eigenvalue of a target Hamiltonian.
We draw connections to predictor-corrector methods for integrating differential equations and numerical continuation methods \cite{allgower2012numerical} by recasting the problem in a differential setting.
We introduce a few other practical numerical techniques (adaptive termination of the classical optimization procedure and resampling) for improving the performance of VQE.

As the name suggests, we borrow notions and ideas from adiabatic quantum computation (AQC) and adiabatic simulation;
see for instance \cite{albash2018adiabatic} for a review.
Improvements to adiabatic simulation for molecular systems was explored recently in \cite{matsuura2021variationally}.
However, the present work focuses on a computational protocol for the enhancement of variational algorithms used on near term gate-based quantum computers.
As we discuss in \Cref{sec:vaqc}, in AQC a quantum state evolves in some Hilbert space via the Schr\"{o}dinger equation.
In contrast, in VAQC a quantum state evolves in some reduced space (due to the ansatz) via a series of applications of VQE.
A similar method was proposed in~\cite{garcia2018addressing}, where the authors introduce the idea of adiabatically assisted VQE. 
The present work expands on this idea, provides different theoretical perspectives, and discusses more connections to other numerical methods.
%the results in \cite{garcia2018addressing} support and  complement ours;
%the examples in \cite{garcia2018addressing} deal with spin chains and an NP-complete Ising problem (exact covering), and demonstrate that a bootstrapping-like method can avoid the local minima that often degrade the performance of VQE.

An approach was recently introduced in \cite{chen2020demonstration} under the name ``adiabatic variational quantum computing.''
This approach is essentially a variational time evolution method, a hybrid quantum-classical method for simulating the Schr\"odinger equation.
However, our approach to VAQC in the present work is fundamentally different;
we do not try to simulate unitary evolution.
We provide a different theoretical perspective to motivate VAQC and to establish certain continuity properties of the optimal ansatz parameters.
Further, this alternate point of view permits us to propose a  numerical method that is different from typical numerical integration methods, and in particular the method proposed in \cite{chen2020demonstration}.

Compared to the quantum approximate optimization algorithm (QAOA) \cite{farhi2014quantum} and other work like \cite{mitarai2019generalization}, the VAQC approach described here permits in practice a wider range of ansatz forms.
Inspired by the adiabatic theorem, the main requirement on the ansatz is that it can represent the ground state of the initial Hamiltonian
(and, critically, the corresponding values of the ansatz parameters are known).
If there is flexibility in the problem to do so, these considerations suggest designing the ansatz and the initial Hamiltonian together to achieve this goal.
In the examples that we consider, we still use ``hardware-efficient'' ans\"atze \cite{kandala_2017}, and suggest a simple modification to how these ans\"atze are typically constructed.
%The rest of the development of VAQC is motivated mostly by classical homotopy continuation methods \cite{allgower2012numerical}.

The VAQC approach and techniques for augmenting VQE are independent of the classical optimization method
(within reason -- ultimately we still need noise-robust optimization methods).
Indeed, we demonstrate that these techniques can practically improve the performance of VQE
(for instance, an order of magnitude or more reduction in error)
using a few different optimization methods (namely stochastic gradient descent, simultaneous perturbation stochastic approximation, and the Nakanishi-Fujii-Todo method).
While \cite{sung2020exploration} notes that the performance of these methods can be sensitive to the hyperparameters used (like step sizes), we expect any hyperparameter tuning would benefit the augmented procedures as well, and only further improve the results we obtain.
Ultimately, the techniques we propose offer similarly or more accurate answers, for the same number of samples.
When considering the number of unique circuits evaluated, these techniques can offer an order of magnitude reduction.

Meta-VQE, recently proposed in \cite{cervera2021meta}, has a similar aim of efficiently calculating the ground state energies of a parameterized Hamiltonian.
However, the approach in meta-VQE is to prepend to the variational ansatz another parameterized circuit which aims to ``learn'' the parametric form of the Hamiltonian.
In this work, we focus on a Hamiltonian depending on a single parameter (``time'') and are less concerned with fitting or learning a functional form of the ground state energy.

As mentioned, VQE has been applied to Ising problems and similar classical combinatorial optimization problems.
Although we do not consider such examples, the present work could also be applied to these problems.
The recent work in \cite{egger2020warm} has a similar goal of improving variational algorithms for combinatorial optimization through warm-start strategies.
Combining those ideas with the present work is a subject for future research.

This work is organized as follows.
In \Cref{sec:vaqc}, we begin with a few theoretical considerations to describe and motivate VAQC.
In \Cref{sec:improvements}, we describe its numerical implementation as a bootstrapping procedure, along with the other practical improvements like resampling and adaptive termination.
In \Cref{sec:experiments}, we discuss some numerical results with simulated quantum devices.
In particular, we show how VAQC can improve on the performance of plain VQE to find a better quality ground-state energy of a small molecule.
Further, we demonstrate that in certain settings, like calculating the Born-Oppenheimer potential energy surface, VAQC and the proposed practical improvements can yield more consistently accurate results.
We then conclude with a few final thoughts.

%%%%%%%%%%%%%%%%%%%%%%%%%%%%%%%%%%%%%%%%%%%%%%%%%%%%%%%%%%%%%%
%%%%%%%%%%%%%%%%%%%%%%%%%%%%%%%%%%%%%%%%%%%%%%%%%%%%%%%%%%%%%%
\section{Variational adiabatic quantum computing}
\label{sec:vaqc}
Abstractly, the goal is to calculate
\begin{equation}
\label{eq:paropt}
    f^*(t) = \min_{\theta}\; f(\theta, t) \equiv \bra{\psi(\theta)} H(t) \ket{\psi(\theta)}
\end{equation}
for a range of values of $t$.
Here, $H$ is a parameterized Hamiltonian while $\psi$ is the VQE ansatz, a parameterized quantum state (trial wavefunction).
For simplicity, we assume that $t$ is a single real value
(typically time),
and for expositional purposes at least, we assume that it is scaled and shifted so $t \in [0,1]$.
%this makes the core idea of the proposed bootstrapping technique clear.
To be concrete, we assume a system of $n$ qubits, so we define the Hilbert space of interest as $\hilbert \equiv (\mbb{C}^2)^{\otimes n}$, and so $H$ is a mapping from $[0,1]$ to the space of self-adjoint linear operators on $\hilbert$, while $\psi$ is a mapping from some domain $\domain$ to $\hilbert$.
We typically think of $\domain$ as a subset of $\mbb{R}^p$
(that is, the quantum state is parameterized by $p$ real parameters),
however we will see that some flexibility is useful.

We may consider Problem~\eqref{eq:paropt} as an approximation of an adiabatic computation;
in this setting
$H(t) = (1 - t)H_I + t H_T$,
where $H_I$ is an initial (or mixing) Hamiltonian, $H_T$ is the target Hamiltonian, and the goal is to determine the ground state of $H_T$.
As another application of \eqref{eq:paropt}, consider calculating the potential energy surface of a molecule, where $H$ is the Hamiltonian parameterized by some molecular coordinate $t$.

At its core, VAQC may be seen as a hybrid quantum-classical homotopy/numerical continuation method, with the goal of approximating the ground state of the parametric Hamiltonian at each instant in time.
At a high level, the same can be said of adiabatic quantum computing (AQC).
In VAQC, the approximation to the ground state is in a reduced space due to the ansatz, and is evolved by a hybrid quantum-classical method (e.g. VQE).
In AQC, however, the approximation to the ground state occurs in the full Hilbert space and evolves via unitary evolution (the Schr\"{o}dinger equation).

The basic requirement of VAQC is that the optimal ansatz parameters $\theta^*$ are continuous as a function of $t$.
Consequently, the solution of Problem~\eqref{eq:paropt} for some value of $t$ will be close to the solution of Problem~\eqref{eq:paropt} for some other nearby value $t'$.
This proves very useful when solving Problem~\eqref{eq:paropt} by local optimization methods like gradient descent.

To establish the required continuity properties, we begin with the following result, stating that there exists a continuous ground state of $H$.
See Appendix~\ref{sec:ce_proof} for its proof.
This result requires a continuously parameterized Hamiltonian with a strictly positive minimum gap between the ground state energy and first excited state energy.
As another, more technical, connection with AQC, similar conditions appear in versions of the adiabatic theorem \cite{jansen2007bounds}.

\begin{proposition}
\label{prop:continuous_eigenvector}
Assume $H$ is continuously differentiable on $[0,1]$.
Assume that there exists real constant $\delta > 0$ such that for each $t$, the lowest and second lowest eigenvalues (of $H(t)$), $\lambda_0(t)$ and $\lambda_1(t)$ respectively, are separated by $\delta$:
$\lambda_0(t) + \delta < \lambda_1(t)$.
Then there exists a continuous mapping $\phi^* : [0,1] \to \hilbert$ such that $\phi^*(t)$ is a ground state of $H(t)$ for all $t \in [0,1]$.
\end{proposition}

The intuition is that if the ansatz $\psi$ can represent the ground state $\phi^*(t)$ for each $t$, and this ground state is continuous, then the optimal ansatz parameters $\theta^*$ \emph{should} be continuous as well.
However, making this precise essentially requires the continuity of $\psi\inv$, the inverse mapping of $\psi$, which is not a property of ans\"atze that typically has been studied.
Since we normally assume that $\psi$ itself also is continuous, this leads to the assumption that $\psi$ is a homeomorphism.
Furthermore, the assumption that the ansatz can represent the ground state is nontrivial as well, but one that is common to most variational methods.
Nevertheless, we state the following result to make it explicit what kinds of assumptions are needed.

\begin{proposition}
\label{prop:continuous_parameters}
Assume that the ansatz $\psi : \domain \to \hilbert_{\psi}$ is a homeomorphism for some subset $\hilbert_{\psi} \subset \hilbert$.
Assume that there exists a continuous mapping $\phi^*$ such that $\phi^*(t) \in \hilbert_{\psi}$ and is a ground state of $H(t)$ for all $t \in [0,1]$.
Then there exists a continuous mapping $\theta^* : [0,1] \to \domain$ such that $\theta^*(t)$ is an optimal solution to Problem~\eqref{eq:paropt} for each $t \in [0,1]$.
\end{proposition}
\begin{proof}
As a homeomorphism, $\psi\inv$ is a continuous mapping $\hilbert_{\psi} \to \domain$.
Thus define $\theta^*(t) = \psi\inv(\phi^*(t))$ and we see that it must be a solution of Problem~\eqref{eq:paropt} and $\theta^*$ is continuous.
\end{proof}

Although we typically think of $\psi$ as a mapping on $\mbb{R}^p$, the abstraction of its domain to the space $\domain$ is useful in the preceding result.
For instance, consider an ansatz whose parameterization is via rotation gates and is periodic with period $2\pi$ with respect to each parameter.
Consequently, as a function on $\mbb{R}^p$, such an ansatz cannot be bijective
(and thus it cannot be a homeomorphism).
If we restrict the parameters to $[0, 2\pi]^p$, then we potentially impose discontinuities on the inverse map.
However, if we define $\domain$ as the quotient space of $\mbb{R}^p$ with respect to the equivalence relation defined by this periodicity%
\footnote{Specifically, the equivalence relation in this example is
$\theta \sim \theta'$
iff
for each $i \in \set{1,\dots,p}$ there exists $k_i \in \mbb{Z}$ such that $\theta_i + k_i 2\pi = \theta_i'$.
}
(along with the corresponding quotient topology),
then we have some hope that this ansatz could be a bijection, and subsequently a homeomorphism.
Then the conclusion of Proposition~\ref{prop:continuous_parameters} is that $\theta^*$ is continuous with respect to this quotient topology on $\domain$;
we can then define%
\footnote{In this case, the quotient map induced by the equivalence relation is a covering map of $\domain$, and so the path lifting property holds;
we can ``lift'' $\theta^*$ to $\vartheta^*$.
See Sections~53 and 54 of \cite{munkres}, in particular Lemma~54.1.
}
a continuous mapping $\vartheta^* : [0,1] \to \mbb{R}^p$ such that $\vartheta^*(t)$ is a solution of Problem~\eqref{eq:paropt} for each $t$.
Practically, this suggests in this setting that we should \emph{not} impose bounds on the parameters when performing the optimization of Problem~\eqref{eq:paropt}.

In AQC, the computation is initialized with a known ground state $\phi^*(0)$ of the initial Hamiltonian $H(0) = H_I$.
Analogously, in VAQC, we require initial ansatz parameters $\theta^0$ such that the ansatz represents the ground state of the initial Hamiltonian:
\[
\psi(\theta^0) = \phi^*(0).
\]
% \[
% \bra{\psi(\theta^0)} H(0) \ket{\psi(\theta^0)} =
%     \min_{\phi} \frac{\bra{\phi} H(0) \ket{\phi}}{\braket{\phi}{\phi}}.
% \]
If $\psi$ is a homeomorphism, then as in the proof of Proposition~\ref{prop:continuous_parameters} we must have 
$\theta^0 = \psi\inv(\phi^*(0)) = \theta^*(0)$.

The requirement that the ansatz be able to represent the ground state at each instant in time may seem too strong, especially if the goal is to find the ground state of the final target Hamiltonian $H_T$ only.
Depending on the properties of the ansatz, a modified analysis is possible.
For instance, if the ansatz spans a linear subspace
(or, more accurately, the range $\psi(\domain)$ of the ansatz equals $\mathcal{L} \cap \mathcal{S}$, where $\mathcal{L}$ is some linear subspace of $\hilbert$ and $\mathcal{S}$ is the unit sphere),
then we are really interested in the spectral gap of the reduced Hamiltonian $A\adjoint H(t) A$, where $A$ is some linear operator from a reduced dimension space $\mbb{C}^d$ to $\hilbert$
(specifically, its columns form an orthonormal basis for $\mathcal{L}$).
Assuming that the reduced parameterized Hamiltonian $t \mapsto A\adjoint H(t) A$ satisfies  conditions analogous to those in Proposition~\ref{prop:continuous_eigenvector}, then there exists a continuous ground state $\phi^*_{\mathcal{L}}$ of the reduced Hamiltonian.
If $\mathcal{L}$ contains the ground state of $H_T$, then once again VAQC provides a method to find it.

Finally, since the motivation behind VAQC is to avoid sub-optimal local minimizers of Problem~\eqref{eq:paropt}, one question is whether the optimal ansatz parameters $\theta^*$ are locally unique, and if so, how close do they come to another local minimizer.
In the numerical implementation of VAQC, we must discretize the parameter $t$ and take finite steps, and so the local uniqueness of the optimal ansatz parameters has implications for the size of the allowable time steps.
While we leave a complete discussion for future work, we state a brief result that is relevant.

\begin{proposition}
\label{prop:unique}
Assume that we can identify $\domain$ with some open subset of $\mbb{R}^p$, $f$ is twice-continuously differentiable, and $\theta^* : [0,1] \to \domain$ is continuous.
Assume that for each $t \in [0,1]$,
$\theta^*(t)$ is a solution of Problem~\eqref{eq:paropt} and
the Hessian matrix $\grad^2_{\theta} f(\theta^*(t), t)$ is positive definite.
Then (for any norm $\norm{\cdot}$) there exists a positive constant $\epsilon$ such that for all $t$, no other local minimizer of $f(\cdot, t)$ comes within a distance less than $\epsilon$ of $\theta^*(t)$.
\end{proposition}
\begin{proof}
%The proof essentially follows that of the inverse function theorem \cite[Thm.~9.24]{rudin}.
Define
$A(t) \equiv \grad^2_{\theta} f(\theta^*(t), t)$.
Then $A$ is positive definite-valued and continuous, and we can find a positive $\eta$ such that for all $t$ we have 
$\eta \le (2 \norm{A(t)\inv})\inv$.
For any $\alpha > 0$, the set 
$\mathcal{K} = \set{(\theta, t) : \norm{\theta - \theta^*(t)} \le \alpha, t\in[0,1]}$
is compact.
% YES - its the image of a compact set under the continuous mapping $(\theta, t) \mapsto (\theta + \theta^*(t), t)$
Thus $\grad^2_{\theta} f$ is uniformly continuous on $\mathcal{K}$, and so we can find an $\epsilon > 0$ such that for all $t$ and $\theta$ with $\norm{\theta - \theta^*(t)} < \epsilon$ we have
$\norm{\grad^2_{\theta} f(\theta,t) - A(t)} < \eta$.
%$\le (2 \norm{A(t)\inv})\inv$.

Now, for each $t$, $\grad_{\theta} f(\theta^*(t),t) = 0$ by the necessary conditions for optimality \cite[Prop.~1.1.1]{bertsekas_nlp}.
For each $t \in [0,1]$ let
$g_{t} : \theta \mapsto \theta - A(t)\inv \grad_{\theta} f(\theta,t)$.
We see that $\grad_{\theta} f(\theta,t) = 0$ if and only if $\theta$ is a fixed point of $g_{t}$.
Then for all $t$, 
%$(\grad g_{t} (\theta))\tr = A(t)\inv (A(t) - \grad^2_{\theta} f(\theta,t))$,
$\grad g_{t} (\theta) = (A(t) - \grad^2_{\theta} f(\theta,t)) A(t)\inv$,
and so
$\norm{\grad g_{t} (\theta)} < \eta \norm{A(t)\inv} \le \sfrac{1}{2}$
for all $\theta$ such that $\norm{\theta - \theta^*(t)} < \epsilon$.
It follows that $g_{t}$ is a contraction mapping on $\mathcal{B}_{\epsilon}(\theta^*(t))$, the open ball of radius $\epsilon$ centered at $\theta^*(t)$.
Thus $g_{t}$ has at most one fixed point in $\mathcal{B}_{\epsilon}(\theta^*(t))$ \cite[Thm.~9.23]{rudin};
we conclude that $\theta^*(t)$ is the only value at which the gradient of $f$ is zero.
Thus, for all $t$, $\theta^*(t)$ is the only local minimizer, and further the only stationary point of $f(\cdot, t)$, in $\mathcal{B}_{\epsilon}(\theta^*(t))$.
\end{proof}

%%%%%%%%%%%%%%%%%%%%%%%%%%%%%%%%%%%%%%%%%%%%%%%%%%%%%%%%%%%%%%
%%%%%%%%%%%%%%%%%%%%%%%%%%%%%%%%%%%%%%%%%%%%%%%%%%%%%%%%%%%%%%
\section{Implementation and practical improvements}
\label{sec:improvements}

In this section, we describe the implementation of VAQC as bootstrapping and the techniques of adaptive termination and resampling in detail.
While the previous section gives precise conditions under which VAQC has some theoretical guarantees, moving forward we will largely view VAQC from a more practical, heuristic perspective.
%%%%%%%%%%%%%%%%%%%%%%%%%%%%%%%%%%%%%%%%%%%%%%%%%%%%%%%%%%%%%%
\subsection{Bootstrapping}
\label{sec:bootstrapping}
Again, recall that the motivation behind VAQC is to overcome the limitations of VQE when faced with a high-dimensional, nonconvex energy landscape with multiple local minima and stationary points.
The proposed bootstrapping procedure aims to navigate this landscape through the solution of a series of related optimization problems, using the solution of the previous problem to warm-start or ``bootstrap'' the solution of the next problem.
Specifically, we assume that we have an ordered sequence of values of $t$, $(t_1, t_2, \dots, t_K)$, at which $f^*$ is wanted.
Then for each $k$, if
$\tilde{\theta}^k$
is an approximate solution of
$\min_{\theta}\, f(\theta, t_k)$,
we use $\tilde{\theta}^k$ as the initial point when solving
$\min_{\theta}\, f(\theta, t_{k+1})$.
Further, as discussed in \Cref{sec:vaqc}, we assume that we have initial optimal parameters
$\theta^0 \in \arg\min_{\theta}\, f(\theta, 0)$
that we use as the initial guess when solving the problem at $t_1$.
Based on Propositions~\ref{prop:continuous_eigenvector} and \ref{prop:continuous_parameters}, we hope that the optimal parameters are continuous, and that $\tilde{\theta}^k$ is a reasonable initial estimate for the next step.

To draw connections to numerical continuation methods \cite{allgower2012numerical}, we discuss VAQC from the point of view of integrating differential equations.
From this perspective, we transform the problem of finding a minimizer of Problem~\eqref{eq:paropt} for each $t$, into the problem of finding a solution of the necessary conditions for optimality:
$\frac{\partial f}{\partial \theta_i}(\theta, t) = 0$
for each $i \in \set{1, \dots, p}$ and each $t \in [0,1]$.
We then differentiate these equations with respect to $t$ and look for a solution $\theta^*$ of the implicit ordinary differential equations
\begin{equation}
    \label{eq:implicit_ode}
    \sum_{j=1}^p \frac{\partial^2 f}{\partial \theta_j \partial \theta_i}(\theta^*(t), t)
    \frac{\partial \theta_j^*}{\partial t} (t) +
    \frac{\partial^2 f}{\partial t \partial \theta_i}( \theta^*(t), t)
     = 0,
	\qquad \forall (i,t) \in \set{1, \dots, p} \times [0,1].
\end{equation}
The initial conditions for the differential equations are set to $\theta^*(0) = \theta^0$.
In general (without the supporting theory of the previous section), a solution of \eqref{eq:implicit_ode} at best gives a \emph{local} minimizer at any $t$.
Similar equations are derived in \cite{mitarai2020theory};
as in that work, we could use an explicit Euler method to numerically integrate Equations~\eqref{eq:implicit_ode}.
Variational (real or imaginary) time evolution also relies on the numerical integration of implicit ordinary differential equations in the space of the ansatz parameters, and in practice often use an explicit Euler method \cite{chen2020demonstration,mcardle2019variational,yuan2019theory}.

The idea of integrating differential equations to find the zero of a system of algebraic equations is the basic idea in numerical continuation methods \cite{allgower2012numerical}.
A common approach there, however, is to use predictor-corrector methods for numerical integration.
The proposed bootstrapping method is akin to using $\tilde{\theta}^k$ as a zeroth-order predictor of the solution at $t_{k+1}$, followed by a local optimization method as the corrector iteration to improve this estimate.

Since the corrector iteration is in effect an application of VQE, it can be expensive, measured in terms of the number of objective function evaluations.
However, an Euler step requires an accurate estimate of the Hessian matrix of $f$, and we will see in practice (see discussion in \Cref{sec:lih}) that the cost of the corrector iteration is comparable to the cost of an explicit Euler step.
Meanwhile, compared to an explicit Euler method, a benefit of the bootstrapping method is that fairly large time steps can be taken,
since errors in the approximation tend to be corrected by the solution of the optimization problem, instead of accumulating as in Euler's method.
Further, depending on the application, the ansatz parameter estimates $\tilde{\theta}^k$ do not need to be highly accurate;
we only need them to be in the ``basin of attraction'' of the global optimal values.
%\footnote{\smh{To be a little more quantitative, if the optimal parameters are the only stationary point in a radius of size $\epsilon$
%(as discussed in Proposition~\ref{prop:unique}),
%then combined with an estimate of the derivative of $\theta^*$ from Equation~\eqref{eq:implicit_ode}
%(or more accurately, a Lipschitz constant)
%we could get an upper bound on the size of the time step that keeps $\theta^*(t_{k+1})$ within $\epsilon$ of $\theta^*(t_{k})$.
%However, the details depend on the convergence properties of the specific corrector iteration and the accuracy of $\tilde\theta^k$.}
%}
For instance, in the setting of an adiabatic computation, we only need the \emph{final} ground state energy with high accuracy, and the VAQC approach is a heuristic to provide us with a high-quality initial point $\tilde{\theta}^K$ that can be used in a final application of VQE to the target Hamiltonian $H_T$.

% Numerical continuation methods \cite{allgower2012numerical} provide insight into more elaborate prediction schemes;
% however, we will see that this sort of zeroth-order predictor already provides a dramatic improvement.

%%%%%%%%%%%%%%%%%%%%%%%%%%%%%%%%%%%%%%%%%%%%%%%%%%%%%%%%%%%%%%
\subsection{Adaptive termination}
Intuitively, if either the ansatz parameters or the objective estimates are varying below some tolerance, then we have determined the objective or the wavefunction as accurately as we need, and there is no need for further optimization.
This is the motivation behind adaptive termination.
This is a standard technique in deterministic optimization, although in the stochastic setting some modifications are necessary.

Since most optimization methods tend to incrementally improve a point, moving toward a minimum, the idea behind adaptive termination is to look at the change in a windowed average of past iterates.
The candidate ansatz parameter values $\theta^k$ and objective estimates $\hat{f}^k$ from the last $N_w$ iterations are saved (for some given value of $N_w$).
The method terminates when we have either
\[
    \frac{1}{N_w}\abs{\sum_{i=1}^{N_w} \hat{f}^{k-i} - \sum_{i=1}^{N_w} \hat{f}^{k+1-i}} \le \varepsilon_f
    \qquad \text{or} \qquad
    \frac{1}{N_w}\norm{\sum_{i=1}^{N_w} \theta^{k-i} - \sum_{i=1}^{N_w} \theta^{k+1-i}} \le \varepsilon_{\theta},
\]
for given tolerances $\varepsilon_f$ and $\varepsilon_{\theta}$.
Evidently the terms in these windowed averages cancel out and the expressions simplify.
In convex stochastic optimization, averaging the iterates is discussed in \cite{polyak1992acceleration}.

Appropriate values of $N_w$, $\varepsilon_f$, and $\varepsilon_{\theta}$ depend on the problem
(for instance, the scaling of the objective $f$ influences the choice of the tolerance $\varepsilon_f$).
In general, smaller values of $\varepsilon_f$ and $\varepsilon_{\theta}$ lead to more accurate results, at the cost of more iterations of the optimization method.
Meanwhile, larger values of $N_w$ can help average out noise in the estimates.

%%%%%%%%%%%%%%%%%%%%%%%%%%%%%%%%%%%%%%%%%%%%%%%%%%%%%%%%%%%%%%
\subsection{Resampling}
Again, the objective function of the optimization problem~\eqref{eq:paropt} is evaluated on a quantum computer, and thus we only have access to a noisy estimator.
For simplicity, we assume this estimator $\hat{f}_m(\theta, t)$ is a sample average estimator consisting of $m$ independent ``samples.''
% \[
%     \hat{f}_m(\theta, t) = \frac{1}{m} \sum_{j=1}^m \Lambda_j(\theta, t)
% \]
The number of circuit evaluations and measurements that are actually required for each sample depends on the quantum hardware and the decomposition of $H(t)$ into basis gates (e.g., as a sum of tensor products of Pauli operators).
We will largely ignore this, and when we refer to a ``circuit evaluation,'' ``shot,'' or ``sample,'' we are referring to some unit of quantum computational resource, $m$ of which is required to calculate $\hat{f}_m(\theta, t)$.
This ignores, for instance, that $H$ may be more efficient to evaluate (e.g. have a simpler decomposition as Pauli operators) for some values of $t$ than others.
However, we feel this abstraction is acceptable, as it makes the present analysis independent of the specific estimation scheme used.

After the optimization/corrector iteration finishes at a parameter value $t$, we have (an approximation of) optimal ansatz parameters $\tilde{\theta}$.
We still have the challenge of accurately estimating the value of $f(\tilde{\theta}, t)$;
the base estimator $\hat{f}_m$ may not be sufficiently accurate.

Fortunately, in many cases we have access to an estimate of the variance of $\hat{f}_m(\tilde{\theta}, t)$.
See for instance \S{}~V.A of Kandala et al.\ \cite{kandala_2017} for the details of the variance estimator in Qiskit \cite{qiskit}.
Assuming this variance is $\sigma^2(\tilde{\theta},t)$, the average of $N_r$ repeated independent estimators has variance $\sfrac{\sigma^2(\tilde{\theta},t)}{N_r}$.
Consequently, we can get an estimate of the objective with standard deviation less than or equal to $\varepsilon_r$ by taking
$N_r > \sfrac{\sigma^2(\tilde{\theta},t)}{\varepsilon_r^2}$.
In fact, considering that $\hat{f}_m$ already consists of $m$ samples, we take
$N_r > \sfrac{m \sigma^2(\tilde{\theta},t)}{\varepsilon_r^2}$
and construct an estimator $\hat{f}_{N_r}(\tilde{\theta}, t)$ with this many samples in the first place.
See also \cite{mcclean2016theory}, which discusses similar ideas as well as an interesting approach from a Bayesian perspective.

However, we reiterate that this proposed resampling procedure is only applied to the final, optimized ansatz parameters.
Another benefit of this construction of the estimator is that it adapts;
based on the variance (which depends on $t$ and $\theta$), we use just enough samples to yield a final objective estimate within a desired tolerance.
Using a noise-robust stochastic optimization method, few samples are needed during optimization when evaluating $\hat{f}_m(\theta, t)$.
But when accuracy matters, at the final evaluation of the optimal objective, we use a high number of samples automatically determined by this resampling procedure.

%%%%%%%%%%%%%%%%%%%%%%%%%%%%%%%%%%%%%%%%%%%%%%%%%%%%%%%%%%%%%%
%%%%%%%%%%%%%%%%%%%%%%%%%%%%%%%%%%%%%%%%%%%%%%%%%%%%%%%%%%%%%%
\section{Experiments}
\label{sec:experiments}
In this section we perform simulations to assess the effectiveness of VAQC and the improvements to VQE.
These are performed with the QasmSimulator in Qiskit v0.15 \cite{qiskit}.
In all experiments, we use $m=64$ samples to build the objective estimator $\hat{f}_m(\theta, t)$.
This is a relatively low number of samples and results in a fair amount of sampling noise.
While not one of our main conclusions, this permits us to assess how noise-robust the various methods are.

The two examples we consider demonstrate two common applications of VQE in quantum chemistry:
calculating a ground state energy for a single molecular configuration, and calculating an energy surface as a function of molecular coordinates.
In the first case, we show how VAQC, with an artificially parameterized Hamiltonian, can yield more accurate energy values than plain VQE.
In the second case, we take advantage of the naturally parametric Hamiltonian and apply VAQC to robustly calculate an energy surface.

%%%%%%%%%%%%%%%%%%%%%%%%%%%%%%%%%%%%%%%%%%%%%%%%%%%%%%%%%%%%%%
\subsection{Optimization methods}
We will use three optimization methods in our experiments: stochastic gradient descent (SGD), simultaneous perturbation stochastic approximation (SPSA), and Nakanishi-Fujii-Todo (NFT).
See Appendix~\ref{sec:optimization_details} for more details.
Each method is well-suited to stochastic settings.
Each method may be augmented with adaptive termination, bootstrapping, resampling, or a combination thereof;
these are denoted with the suffixes ``A'', ``B'', or ``R'', respectively.
For instance, ``NFT-BR'' refers to NFT with bootstrapping and resampling applied.

There are many other optimization methods that can be applied in this context;
see \cite{lavrijsen2020classical,sung2020exploration} for more extensive discussion along these lines.
Again, our intention is not to exhaustively compare or benchmark optimization methods, but rather to assess the effect of VAQC and the other augmentations.
The three methods are meant to be representative of a few different classes of optimization methods:
strongly local, gradient-based methods (SGD);
derivative-free methods with some random search (SPSA);
and domain-specific methods with some guarantee of finding a global minimum (NFT).

%%%%%%%%%%%%%%%%%%%%%%%%%%%%%%%%%%%%%%%%%%%%%%%%%%%%%%%%%%%%%%
\subsection{VAQC calculation of ground state of lithium hydride}
\label{sec:lih}
We use the VAQC approach to robustly compute the ground state energy, and compare this with VQE.
The target Hamiltonian $H_T$ acting on four qubits comes from the electronic energy of $\mathrm{LiH}$ when the distance between the nuclei is 2.5\AA.
This is not the equilibrium bond length;
we specifically choose this internuclear distance since it can be a challenging point for VQE to find a ground state energy.
See Appendix~\ref{sec:hamiltonian} for its full specification.
Errors are with respect to the full configuration interaction energy
(i.e., the lowest eigenvalue of the Hamiltonian, calculated with classical methods).
The objective function is measured in units of Hartree, and so we choose the tolerances with the aim of achieving an accuracy of order $10^{-3}$.
Thus, for instance, we will use $\varepsilon_r = 5 \times 10^{-4}$ in the resampling procedure, whenever it is used.

\subsubsection*{Procedure}
We use a four-layer RY ansatz with a few additional gates appended to the end of the circuit.
This results in an ansatz with $p = 20$ parameters.
(see \Cref{fig:lih_circuit}).
Note that $\psi(0) = \ket{0110}$, since $R_Y(0) = I$, each CNOT does not modify the initial state $\ket{0000}$, and the final $X$ gates flip the second and third qubits.
We define the initial Hamiltonian as
\[
H_I =
  (0.568) I^{\tensor 4}
+(-0.102) Z_1
+ (0.245) Z_2
+ (0.102) Z_3
+(-0.245) Z_4,
\]
% -0.245 Z\tensor I\tensor I\tensor I +
%  0.102 I\tensor Z\tensor I\tensor I +
%  0.245 I\tensor I\tensor Z\tensor I +
% -0.102 I\tensor I\tensor I\tensor Z +
%  0.568 I\tensor I\tensor I\tensor I
which indeed has $\psi(0)=\ket{0110}$ as a ground state.
The motivation for this form of $H_I$ is that the target Hamiltonian as a weighted sum of tensor products of Pauli operators has the form
$
H_T = H_I + (\text{other terms}),
$
where ``other terms'' captures the rest of the terms that do not include $Z_1$, $Z_2$, $Z_3$, $Z_4$, or $I^{\tensor 4}$.
Furthermore, all other terms have a weight with absolute value less than $0.2$, and so $H_I$ is capturing some of the more significant terms.
While the Hartree-Fock state and corresponding Hamiltonian are another option for constructing the ansatz and initial Hamiltonian, the present choice demonstrates that even a mildly problem-specific heuristic for constructing the initial Hamiltonian is successful.
%TODO: chemistry reference?

\begin{figure}
\centering
\begin{quantikz}
\lstick{$q_1:$ \ket{0}} &\gate{R_Y(\theta_1)}  &\ctrl{1}\gategroup[wires=4,steps=4, style={dashed,rounded corners}]{repeat $4\times$}
                                                         &\qw      &\qw      &\gate{R_Y(\theta_{4r+1})} &\qw     &\qw\\
\lstick{$q_2:$ \ket{0}} &\gate{R_Y(\theta_2)}  &\gate{X} &\ctrl{1} &\qw      &\gate{R_Y(\theta_{4r+2})} &\gate{X}&\qw\\
\lstick{$q_3:$ \ket{0}} &\gate{R_Y(\theta_3)}  &\qw      &\gate{X} &\ctrl{1} &\gate{R_Y(\theta_{4r+3})} &\gate{X}&\qw\\
\lstick{$q_4:$ \ket{0}} &\gate{R_Y(\theta_4)}  &\qw      &\qw      &\gate{X} &\gate{R_Y(\theta_{4r+4})} &\qw     &\qw
\end{quantikz}
\caption{Ansatz for $\mathrm{LiH}$ ground state study.}
\label{fig:lih_circuit}
\end{figure}
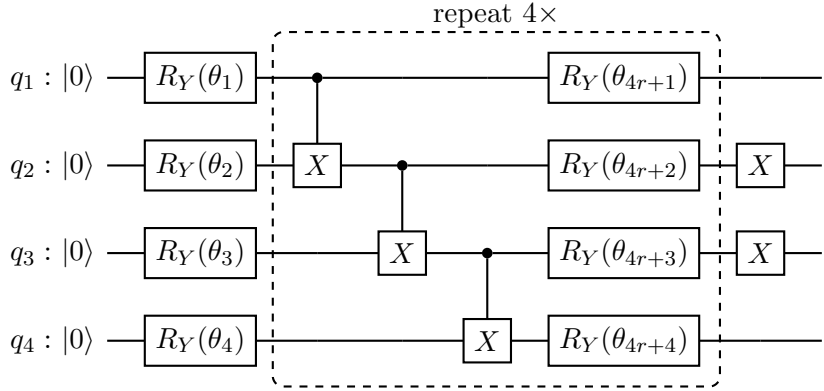

Subsequently, we define the parameterized Hamiltonian:
$H(t) = (1 - t) H_I + t H_T$.
The ``time'' points are taken to be $t_k = 1 - (1 - 0.05 k)^3$ for $k \in \set{1,\dots, 20}$;
this is akin to using a dimensionless time step of $0.05$, along with a schedule of $s : t \mapsto 1 - (1-t)^3$.
While the term ``schedule'' is borrowed from adiabatic computation, the motivation behind its form in the present setting is that we would like to approach the target Hamiltonian gradually, without taking any large steps that might inhibit the convergence of the optimization-based corrector iteration.
Consequently, for a fixed number of time steps, we posit that they should be distributed so as to avoid large perturbations near the target Hamiltonian
(this motivates a schedule of the form $t \mapsto 1-(1-t)^a$ with $a>1$).
This should also help with the relatively slow convergence of SGD, which we use for the corrector iteration.

The corrector iteration in specific is SGD-AB, with a maximum of $2p = 40$ iterations, $N_w = 10$, $\varepsilon_f = 10^{-3}$, and $\varepsilon_{\theta} = 10^{-3}$.
SGD requires $2p + 1 = 41$ objective estimates per iteration (using the parameter shift rule to estimate the gradient).
Thus, each corrector iteration requires up to $2p(2p + 1)$ objective estimates.
This is roughly the same amount of effort required to estimate the Hessian matrix of $f$, which would be required for an Euler method applied to the differential equation~\eqref{eq:implicit_ode}:
each of the $\sfrac{p(p+1)}{2}$ unique entries of the symmetric Hessian require up to four objective estimates via an iterated parameter-shift rule to get a second derivative
(although see \cite{mitarai2020theory} for an alternative approach that requires fewer estimates if an ancilla qubit is available).
However, in our approach, the adaptive termination means that we typically use fewer than the $40$ allowed iterations.

With these settings, VAQC gets close to a good solution, but requires some polishing;
consequently, at the last time point $t_{20} = 1$, we apply SGD (with adaptive termination, bootstrapping, and resampling) with a maximum of $80$ iterations, $N_w = 10$, $\varepsilon_f = 10^{-4}$, and $\varepsilon_{\theta} = 5 \times 10^{-4}$.

VAQC will require more computation than ``plain'' VQE applied directly to the target Hamiltonian $H_T$.
Given a budget of iterations, the best we can do with VQE is to repeatedly apply it with different randomly chosen initial points.
Thus, we compare against plain VQE, using SGD-AR and the same ansatz, but with initial ansatz parameters uniformly randomly chosen from $[-\pi,\pi]^p$.
%We set a maximum of $19 \times 40 + 80$ iterations (the same total maximum for the VAQC procedure).
We take $N_w = 10$, $\varepsilon_f = 10^{-4}$, $\varepsilon_{\theta} = 5 \times 10^{-4}$
(the same values used for the final time point of VAQC).

\subsubsection*{Results}
VAQC robustly produces a high-quality solution;
over five independent trials, VAQC yields a final objective value within $2.5 \times 10^{-3}$ Hartree of the exact energy.
See \Cref{tab:vaqc}.
Over these five repetitions, VAQC requires $1913$ total iterations of SGD.
Giving this budget of iterations to plain VQE, we are able to sample $33$ different initial ansatz parameters, with an average of $58$ iterations required to converge VQE.
However, even these multiple applications of VQE with different initial ansatz parameters do not produce as accurate an answer;
the minimum absolute error that is achieved is $9.5 \times 10^{-3}$ Hartree.
Furthermore, the results are not nearly as consistent, with errors as large as $0.1641$ Hartree.
This supports the conclusion that the VAQC procedure is capable of navigating a non-trivial landscape in order to improve upon the performance of plain VQE.

\begin{table}[h]
\caption{Accuracy of VAQC and VQE for $\mathrm{LiH}$ ground state energy}
\label{tab:vaqc}
\centering
\begin{tabular}{ccc}\hline
    Method  & Mean absolute error (Ha)   & Min / Max absolute error (Ha) \\
\hline
    VAQC    & 0.0016                     & 0.0004 / 0.0025 \\ 
    VQE     & 0.0640                     & 0.0095 / 0.1641 \\
\hline
\end{tabular}
\end{table}

% TOTAL ITERATIONS: 1913 over 5 runs, 20 steps per run = 19 iterations per step
VAQC requires on average $19.13$ iterations per time step, or $19.13 \times 41 = 784.33$ function evaluations per time step.
As noted, an explicit Euler step would require at least $4(\sfrac{p(p+1)}{2}) = 840$ function evaluations per time step (if an ancilla qubit is not available).
Thus, we see that VAQC (with SGD, at least) has a per time step cost comparable to an explicit Euler scheme.

\subsection{Potential energy surface of molecular hydrogen}
In this problem, our goal is to calculate an accurate Born-Oppenheimer potential energy surface for $\mathrm{H}_2$.
We will use this as an opportunity to assess the impact of the various proposals (bootstrapping/VAQC, adaptive termination, and resampling).

The parameter $t$ corresponds to internuclear distance, and we seek the electronic energy at the points
$\set{0.4, 0.55, 0.7, 0.85, 1.2, 1.6, 1.9, 2.2, 2.6, 3.0, 4.0, 5.0}$ (\AA).
%Based on experience and intuition, we do not expect the potential energy surface to vary much for distances greater than $3$\AA, which explains the sparsity of the points in that region.
We construct the Hamiltonian using PySCF \cite{pyscf} and the STO-3G basis set \cite{hehre1969self,collins1976self}, following the procedure outlined in \cite{kandala_2017}.
This Hamiltonian is then transformed into the particle/hole representation \cite{barkoutsos_2018} and mapped to a qubit Hamiltonian using parity mapping \cite{bravyi_2002, seeley_2012}, which results in a two qubit representation.
We use a one-layer RY ansatz with additional gates prepended that prepare the Hartree-Fock state%
\footnote{As discussed in the lithium hydride study 
(\Cref{sec:lih}), arguably a better ansatz would be one with the gates that prepare the Hartree-Fock state appended to the \emph{end} of the circuit.
In fact, the methods considered do perform better with this ansatz, but it subsequently makes it harder to see the effects of the proposed improvements like bootstrapping.}.
See \Cref{fig:h2_circuit}.
Again, errors are with respect to the full configuration interaction energy
(calculated classically).
We target an accuracy level of $10^{-3}$ Hartree.
Thus, for instance, we will use $\varepsilon_r = 5 \times 10^{-4}$ in the resampling procedure, whenever it is used.
While in general, lower error is better, we are more interested in consistently low error, and further, a method that achieves error much lower than $10^{-3}$ could be seen as wasting effort.

% The accuracy required depends on the ultimate application;
% here, we take the desired accuracy to be $10^{-3}$ Hartree.
% Thus, in our results we take $N_w = 10$, $\varepsilon_f = 10^{-4}$ (Hartree), and $\varepsilon_{\theta} = 5 \times 10^{-4}$ for the adaptive termination method, while using $\varepsilon_r = 5 \times 10^{-4}$ in the resampling procedure.

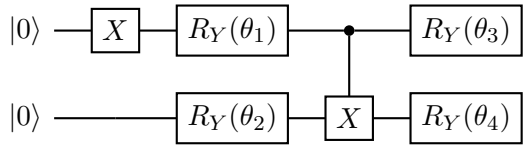
\begin{figure}
\centering
\begin{quantikz}
    \lstick{\ket{0}} & \gate{X} & \gate{R_Y(\theta_1)} & \ctrl{1} & \gate{R_Y(\theta_3)} \\
    \lstick{\ket{0}} & \qw      & \gate{R_Y(\theta_2)} & \gate{X} & \gate{R_Y(\theta_4)}
\end{quantikz}
\caption{Ansatz for $\mathrm{H}_2$ potential energy surface study.}
\label{fig:h2_circuit}
% See
% https://qiskit.org/documentation/stubs/qiskit.circuit.library.U3Gate.html
% for U3 gate
% output from qiskit in log:
%      ┌─────────────┐┌──────────┐     ┌──────────┐
% q_0: ┤ U3(pi,0,pi) ├┤ RY(θ[0]) ├──■──┤ RY(θ[2]) ├
%      └─┬──────────┬┘└──────────┘┌─┴─┐├──────────┤
% q_1: ──┤ RY(θ[1]) ├─────────────┤ X ├┤ RY(θ[3]) ├
%        └──────────┘             └───┘└──────────┘
% note U3(pi,0,pi) is an X gate
\end{figure}

\subsubsection*{Procedure}
For any method using adaptive termination, we take $N_w = 10$, $\varepsilon_f = 10^{-4}$, and $\varepsilon_{\theta} = 5 \times 10^{-4}$.
%For any method using resampling, we take $\varepsilon_r = 5 \times 10^{-4}$.
Any method using bootstrapping begins at the smallest distance and proceeds in order of increasing distance.
For any method or iteration that is not bootstrapped, the initial guess for the ansatz parameters is zero.
While $\psi(0)$ is not exactly the ground state for the initial Hamiltonian at $0.4$\AA, it is a reasonable starting point.

For each optimization method, we at least look at the performance of the baseline optimizer, optimizer plus resampling, and optimizer plus resampling plus bootstrapping.
In the particular case of SGD, we also look at SGD plus resampling, bootstrapping, and adaptive termination
(for NFT and SPSA we use the current implementations in Qiskit, which would have to be modified to allow adaptive termination).

To assess the robustness of the methods, we perform five independent repetitions of the potential energy surface calculation for each method.
To be as fair as possible in our comparisons of the various methods, we make sure that each method uses approximately the same number of total samples in the calculation of the potential energy surface, over all the independent repetitions.
Since it is only SGD-ABR with its adaptive termination for which this budget of samples is not known beforehand, we run this method first and allocate its effort to the others.
While we do not directly control how many samples are used for resampling, two methods that use resampling will use approximately the same total number of samples during the resampling phase.

Thus we run SGD-ABR first, calculate the total number of samples used, and set an iteration limit (which is the same for each internuclear distance) for SGD, SPSA, and NFT so that they approximately use the same total number of samples
(taking into account the different number of function calls per iteration for each method).
Meanwhile, we take the number of samples used only during the optimization phase from SGD-ABR, and again set an iteration limit for SGD-R, SGD-BR, SPSA-R, SPSA-BR, NFT-R, and NFT-BR, so that the optimization phases of each method use the same number of samples
(while the number of samples used in the resampling phases are roughly similar).

\subsubsection*{Improvements in accuracy}

%\begin{figure}
%\centering
%\includegraphics[width=0.33\textwidth]{boxplot-errors-spsa.png}~
%\includegraphics[width=0.33\textwidth]{boxplot-errors-nft.png}~
%\includegraphics[width=0.33\textwidth]{boxplot-errors-sgd.png}
%\caption{Box and whisker plot for absolute error (over independent repetitions and internuclear distances) for various methods for $\mathrm{H}_2$ potential energy surface.
%Boxes represent middle quartiles, whiskers represent minimum to maximum, and the central line is the mean.}
%\label{fig:boxplot}
%\end{figure}

\begin{figure}
\centering
\includegraphics[width=\textwidth]{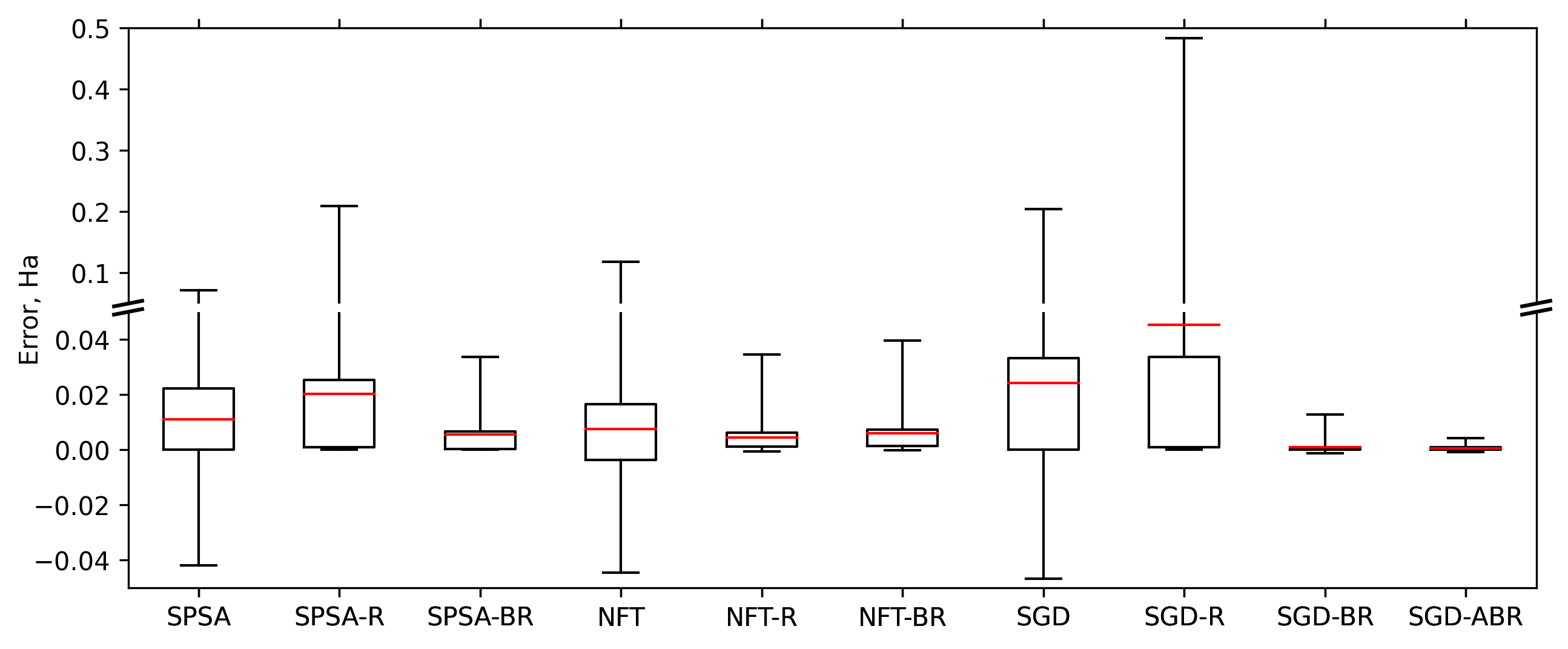}
\caption{Box and whisker plot for error (over independent repetitions and internuclear distances) for various methods for $\mathrm{H}_2$ potential energy surface.
Boxes represent middle quartiles, whiskers represent minimum to maximum, and the central line is the mean.
Note the change in scaling at $0.05$ Hartree on the y-axis.}
\label{fig:boxplot}
\end{figure}

\begin{figure}
\centering
\includegraphics[width=0.5\textwidth]{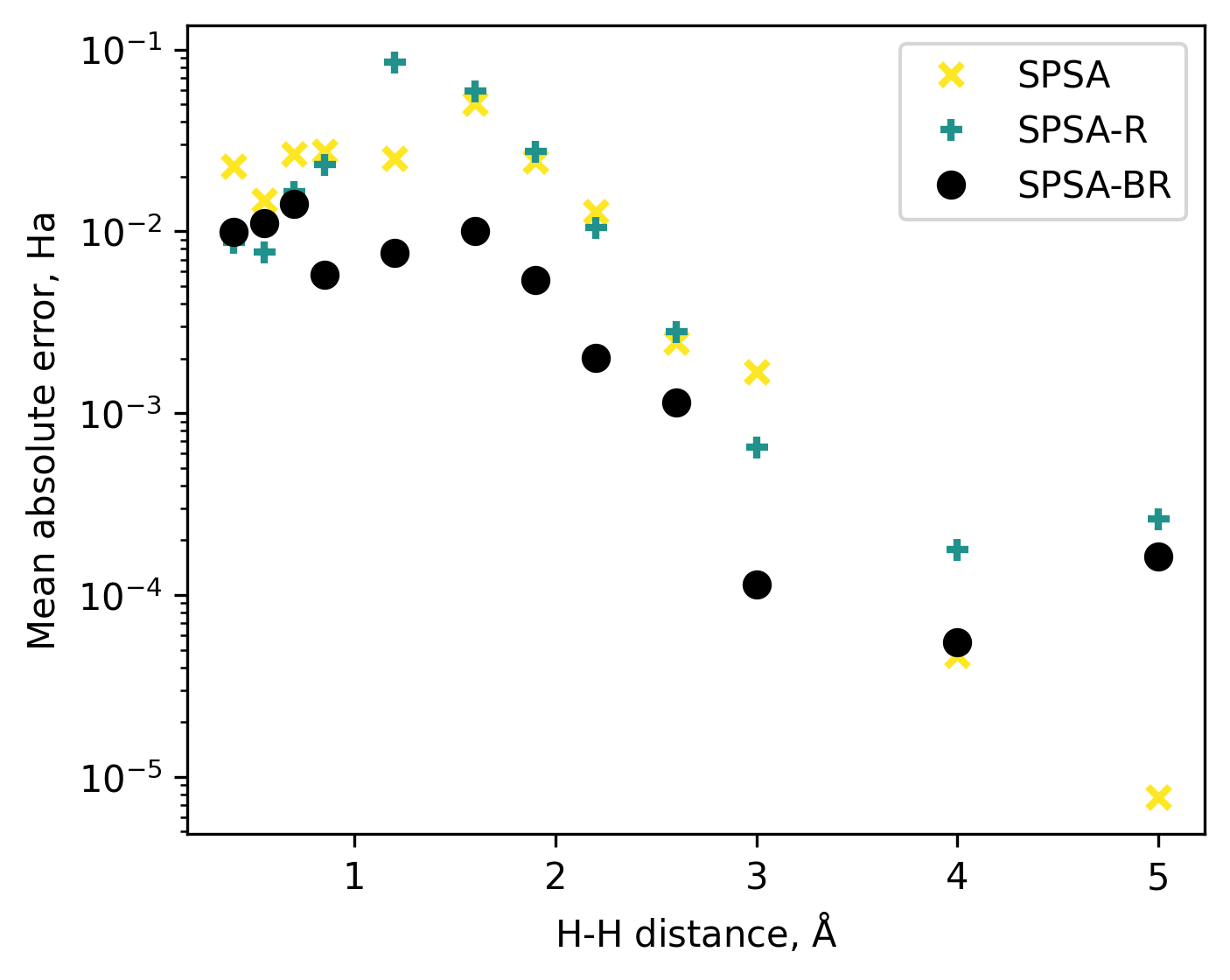}~
\includegraphics[width=0.5\textwidth]{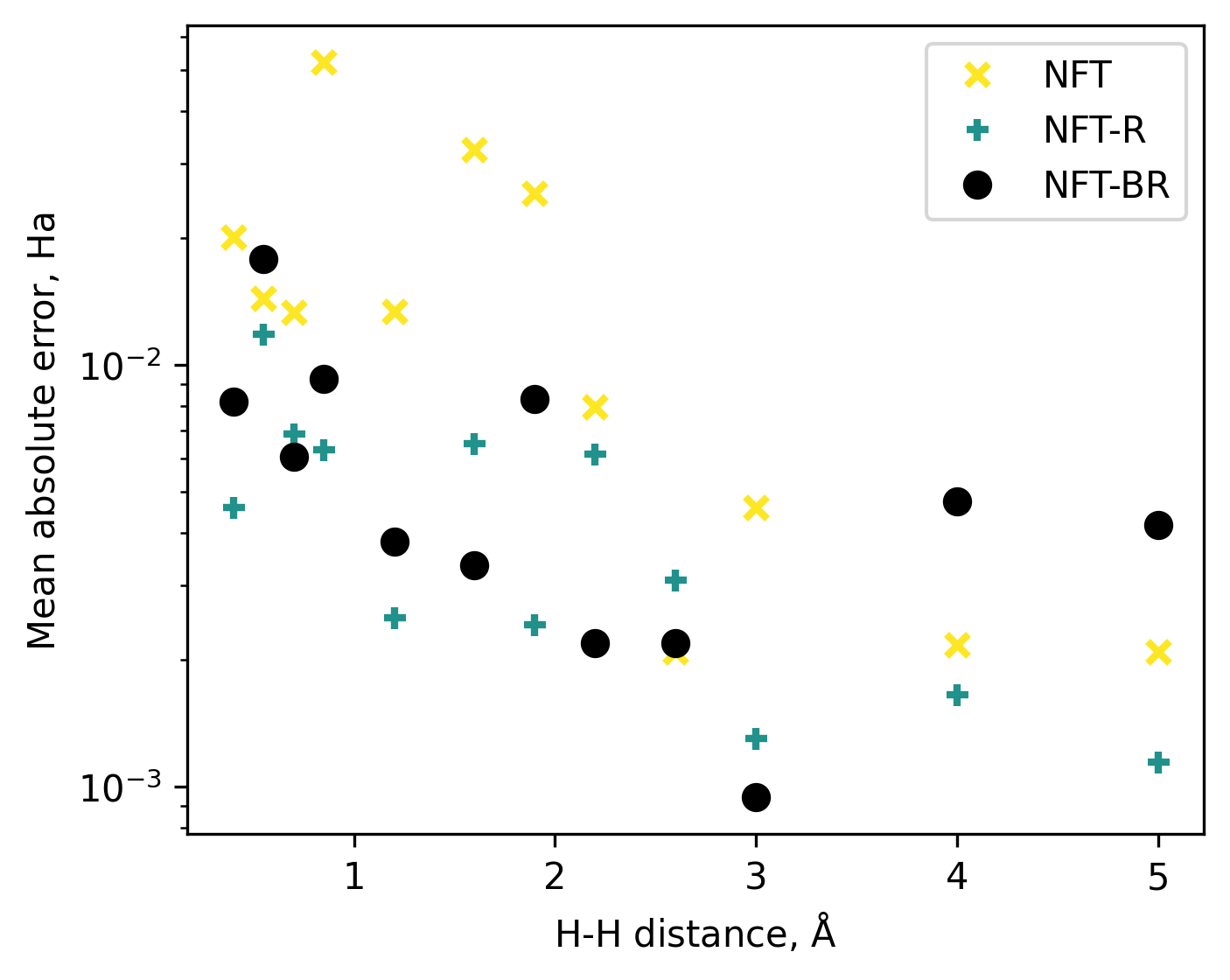}
\caption{Mean absolute error (over five independent repetitions for each internuclear distance) versus distance.
We have SPSA (left) and NFT (right) along with the corresponding augmented versions.}
\label{fig:ae-vs-dist}
\end{figure}

\begin{figure}
\centering
\includegraphics[width=0.7\textwidth]{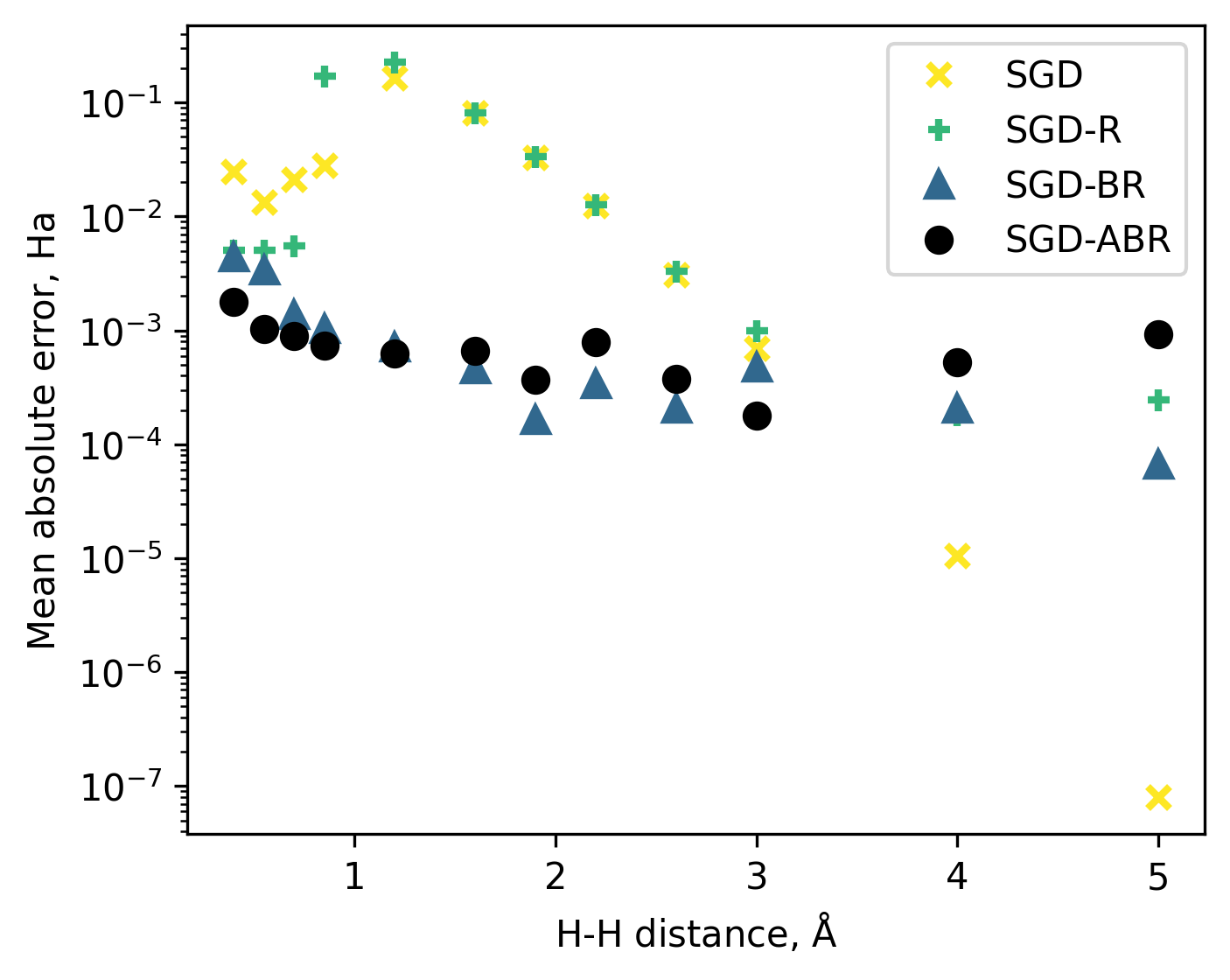}
\caption{Mean absolute error (over five independent repetitions for each internuclear distance) versus distance for SGD and its augmented versions.}
\label{fig:ae-vs-dist-sgd}
\end{figure}

\Cref{fig:boxplot} visualizes the errors, over all independent repetitions and internuclear distances, for each baseline optimization method and its augmented versions.
Meanwhile, \Cref{fig:ae-vs-dist,fig:ae-vs-dist-sgd} show the progression of mean absolute error at each internuclear distance for each method.
In general, where there is error, we must distinguish between sampling error and failure of the optimization method to find optimal ansatz parameters.
The resampled versions of the optimization methods help make this distinction.
For instance, for SPSA and SGD, the addition of resampling does not make significant improvements
(besides eliminating negative errors, which must be due to sampling error).
This indicates that most of the error for those methods is due to poor quality ansatz parameters.
Meanwhile, SPSA-BR improves a lot over SPSA and SPSA-R, and in particular is nearly an order of magnitude more accurate in the region of $1$ to $2$\AA, indicating that bootstrapping yields a significant improvement.
Similarly, we see that SGD improves significantly with bootstrapping and resampling.
Looking at \Cref{fig:ae-vs-dist-sgd}, SGD-BR and SGD-ABR consistently achieve the desired threshold of $10^{-3}$ Hartree across the internuclear distances.
Furthermore, SGD-BR and SGD-ABR are not limited to an accuracy of $10^{-3}$ Hartree;
with a tighter tolerance $\varepsilon_r$ and more optimization iterations we expect them to improve.

While mean absolute error appears similar between SGD-BR and SGD-ABR, \emph{maximum} absolute error over all the distances and independent repetitions is $0.0128$ Hartree for SGD-BR, and $0.0043$ Hartree for SGD-ABR.
See also \Cref{tab:ae-vs-evals}.
This is where adaptive termination proves valuable;
SGD-ABR expends more effort when necessary in order to get closer to a desired accuracy level.

\Cref{fig:boxplot,fig:ae-vs-dist} show that NFT is already a fairly robust optimization method, with resampling appearing to make most of the improvement over the base optimizer.
This makes sense, since NFT can take large steps when updating the ansatz parameters.
Thus, the final parameters depend less on the initial values given to NFT, making bootstrapping less effective.

\subsubsection*{Reduction in unique circuits}
While we have kept the total number of samples the same in our comparisons, the number of \emph{unique} circuits that are executed is an interesting metric to consider.
In many hardware architectures, the overhead of preparing a circuit can be significant, and so once prepared, executing the same circuit many times is relatively efficient.
So, consider \Cref{tab:ae-vs-evals}.
We see that, compared to the base optimizers, adding bootstrapping and resampling results in better mean and maximum absolute error, as well as an order of magnitude fewer unique circuit evaluations.
This is because while resampling contributes many objective evaluations, they are all at the same ansatz parameter values.

\begin{table}
\caption{Accuracy and cost for various methods for $\mathrm{H}_2$ potential energy surface.
Unique circuit evaluations are estimated from the number of function evaluations during optimization.}
\label{tab:ae-vs-evals}
\centering
\begin{tabular}{lccp{3.5cm}}\hline
    Method  & Mean absolute error (Ha)   & Max absolute error (Ha) & Mean unique circuit evaluations \\
\hline
    SPSA    & 0.0174    & 0.0721    & 2341 \\
    SPSA-BR & 0.0056    & 0.0337    & 205 \\
    NFT     & 0.0158    & 0.1180    & 2340 \\
    NFT-BR  & 0.0059    & 0.0397    & 204 \\
    SGD     & 0.0320    & 0.2040    & 2340 \\
    SGD-BR  & 0.0011    & 0.0128    & 198 \\
    SGD-ABR & 0.0007    & 0.0043    & 203 \\
\hline
\end{tabular}
\end{table}

%%%%%%%%%%%%%%%%%%%%%%%%%%%%%%%%%%%%%%%%%%%%%%%%%%%%%%%%%%%%%%
%%%%%%%%%%%%%%%%%%%%%%%%%%%%%%%%%%%%%%%%%%%%%%%%%%%%%%%%%%%%%%
\section{Conclusions}

We have presented the concept of variational adiabatic quantum computing (VAQC), its practical implementation as a bootstrapping procedure, and some other practical numerical techniques to improve VQE.
When applied to the calculation of the ground state energies of simple molecules, these techniques yield consistently more accurate results compared to a more basic implementation of VQE.
This is particularly evident when calculating a potential energy surface, where a series of high accuracy energies are needed.

We have established some basic theoretical guarantees for the performance of VAQC.
While the conditions under which these guarantees hold are difficult to verify, VAQC performs quite well as a heuristic for improving VQE.
We expect the VAQC approach to benefit from further research in AQC;
analysis of the spectral gaps of Hamiltonians of interest should impact the approach discussed here, and homotopy maps could inspire novel ans\"atze.
Meanwhile, VAQC may be a useful tool for better understanding or simulating adiabatic computations.
%%%%%%%%%%%%%%%%%%%%%%%%%%%%%%%%%%%%%%%%%%%%%%%%%%%%%%%%%%%%%%
%%%%%%%%%%%%%%%%%%%%%%%%%%%%%%%%%%%%%%%%%%%%%%%%%%%%%%%%%%%%%%
\subsection*{Acknowledgements}
The authors thank Bryce Fuller for discussions on the manuscript and Michael Andrews for help on some of the topology concepts.
%%%%%%%%%%%%%%%%%%%%%%%%%%%%%%%%%%%%%%%%%%%%%%%%%%%%%%%%%%%%%%
%%%%%%%%%%%%%%%%%%%%%%%%%%%%%%%%%%%%%%%%%%%%%%%%%%%%%%%%%%%%%%
\appendix
\section{Proof of Proposition~\ref{prop:continuous_eigenvector}}
\label{sec:ce_proof}

The following proof essentially elaborates on Remarks~5.10 and 4.3 in \cite[Chapter II]{kato}.
A fundamental tool in that analysis is the resolvent operator
\[
    R(z,t) = (H(t) - zI)\inv
\]
where $z$ is a complex number and $I$ is the identity on $\hilbert$.
Evidently the resolvent is only defined for $(z,t)$ such that $z$ is not an eigenvalue of $H(t)$.
So, let $\mathcal{D}_R$ be the preimage under $(z,t) \mapsto H(t) - zI$ of the set of invertible operators.
Then, as the preimage under a continuous mapping of an open set, $\mathcal{D}_R$ is an open subset of $\mathbb{C} \times [0,1]$.
Further, $R$ is well-defined on $\mathcal{D}_R$.
Since $H$ and the inversion of an operator are continuous and further, differentiable, $R$ is continuous and differentiable on $\mathcal{D}_R$ with derivative
$\frac{\partial R}{\partial t}(z,t) = -R(z,t) \dot{H}(t) R(z,t)$.
This shows that $\frac{\partial R}{\partial t}$ is continuous on $\mathcal{D}_R$.
Importantly, for any $s$,
$\frac{R(z,t) - R(z,s)}{t - s} \to \frac{\partial R}{\partial t}(z,s)$ pointwise in $z$ as $t \to s$, and the convergence is uniform on any compact subset of $\mathbb{C}$ not containing an eigenvalue of $H(s)$.%
\footnote{A little more detail here:
if $R_{\delta}$ is defined on $\mathcal{D}_R$ as the difference quotient when $t \neq s$, and as $\frac{\partial R}{\partial t}(z,s)$ for $t=s$, then it is a continuous function on $\mathcal{D}_R$.
Thus $R_{\delta}(\cdot, t)$ must converge uniformly to $R_{\delta}(\cdot, s) = \frac{\partial R}{\partial t}(\cdot,s)$ on compact subsets $K$ of $\mathbb{C}$ such that $K \times (s - \epsilon, s+\epsilon) \subset \mathcal{D}_R$, for some positive real $\epsilon$.
}
That $\mathcal{D}_R$ is open is critical here;
if $K$ is a compact subset of the complex plane which does not contain an eigenvalue of $H(s)$, then $K \times \set{s}$ is a subset of $\mathcal{D}_R$, and so $K \times \set{t}$, for $t$ near $s$, is also a subset of $\mathcal{D}_R$.

With this, we can define an eigenprojection $P$ as
\[
P = -\frac{1}{2 \pi i} \oint_{\Gamma} R(z,s) dz.
\]
Here, $P$ is the eigenprojection corresponding to all eigenvalues of $H(s)$ enclosed by the positively-oriented closed curve $\Gamma$ in the complex plane.
This characterization of the eigenprojection also appears in proofs of the adiabatic theorem;
see for instance \cite{jansen2007bounds}.
Another basic fact we will use is that the eigenvalues of $H$ may be chosen as continuous functions on $[0,1]$, since $H$ is continuous and the eigenvalues are real.
This means that for $t$ close to $s$, $\Gamma$ encloses the ``same'' eigenvalues.

\begin{proposition*}
Assume $H$ is continuously differentiable on $[0,1]$.
Assume that there exists real constant $\delta > 0$ such that for each $t$, the lowest and second lowest eigenvalues (of $H(t)$), $\lambda_0(t)$ and $\lambda_1(t)$ respectively, are separated by $\delta$:
$\lambda_0(t) + \delta < \lambda_1(t)$.
Then there exists a continuous mapping $\phi^* : [0,1] \to \hilbert$ such that $\phi^*(t)$ is a ground state of $H(t)$ for all $t \in [0,1]$.
\end{proposition*}
\begin{proof}
To start, we show that we have a continuously differentiable eigenprojection.
Fix $s \in [0,1]$.
We may choose a closed curve $\Gamma$ enclosing the lowest eigenvalue $\lambda_0(s)$, and not enclosing or passing through any other eigenvalue
(this is possible by the assumption that the lowest eigenvalue is separated by $\delta$).
Define
\[
    P_0(s) = -\frac{1}{2 \pi i} \oint_{\Gamma} R(z, s) dz.
\]
Consequently, $P_0(s)$ is the eigenprojection corresponding to $\lambda_0(s)$.
Critically, this expression continues to hold for all $t$ in a neighborhood of $s$ by the continuity of $\lambda_0$.
Then note that
\[
    \frac{P_0(t) - P_0(s)}{t - s} = -\frac{1}{2 \pi i} \oint_{\Gamma} \frac{R(z,t) - R(z,s)}{t - s} dz
\]
for $t$ sufficiently close to $s$.
Taking the limit as $t$ approaches $s$, we obtain
$\dot{P}_0(s) = -\sfrac{1}{2 \pi i} \oint_{\Gamma} \frac{\partial R}{\partial s}(z,s) dz$
by the uniform convergence of $\frac{R(z,t) - R(z,s)}{t - s}$ mentioned above.
Once again, this expression continues to hold for $t$ near $s$, so $P_0$ is continuously differentiable on a neighborhood of $s$, and since $s$ was arbitrary, it is continuously differentiable on $[0,1]$.

Next we show that there exists a continuous $U$ such that for all $t$,
$U(t)$ is invertible and
$P_0(t) = U(t) P_0(0) U\inv(t)$.
Since $P_0$ is a projection,
$(P_0(t))^2 = P_0(t)$
for all $t$, so differentiating yields
\begin{equation}
\label{eq:one}
P_0 \dot{P}_0 + \dot{P}_0 P_0 = \dot{P}_0;
\end{equation}
multiplying from the right and left by $P_0$ gives
\begin{equation}
\label{eq:two}
P_0 \dot{P}_0 P_0 = 0.
\end{equation}
Then let
$Q(t) = \dot{P}_0(t) P_0(t) - P_0(t) \dot{P}_0(t)$
be the commutator of $\dot{P}_0$ and $P_0$.
We have that
\begin{equation}
\label{eq:three}
\dot{P}_0  = QP_0 - P_0 Q
\end{equation}
using the definition of $Q$, \eqref{eq:one}, and \eqref{eq:two}.
We construct $U$ as the solution of the initial value problem in linear ordinary differential equations
\begin{equation}
\label{eq:lode}
    \dot{V}(t) = Q(t) V(t), \forall t \in [0,1], \qquad V(0) = V_0.
\end{equation}
In particular, for the initial value $V_0 = I$, we define the solution as $U$.
We note that a general solution of \eqref{eq:lode} may be given by $V(t) = U(t)V_0$.
The inverse of $U$ is the solution of the initial value problem
\[
    \dot{W}(t) = -W(t) Q(t), \forall t \in [0,1], \qquad W(0) = I.
\]
For a solution $W$ of the above, we see that
$\frac{d (WU)}{dt} = W\dot{U} + \dot{W}U = WQU - WQU = 0$;
thus $W(t)U(t)$ is constant and must equal its initial value $W(0)U(0) = I$ for all $t$.
This shows that $U$ is invertible.
Then note that $P_0 U$ is a solution of \eqref{eq:lode} with initial value $P_0(0)$:
\[
    \frac{d (P_0 U)}{dt} = \dot{P}_0 U + P_0 \dot{U} = (\dot{P}_0 + P_0 Q) U = Q (P_0 U)
\]
where we have used \eqref{eq:three} in the last equality.
Since the solution of \eqref{eq:lode} is unique for a given initial value, $P_0(t) U(t)$ must coincide with the general solution noted above $U(t) P_0(0)$.
In other words
$P_0(t) U(t) = U(t) P_0(0)$
for all $t$;
thus we have
$P_0(t) = U(t) P_0(0) U\inv(t)$
as desired.

Finally, let $\phi^0$ be any (non-zero) vector in the eigenspace of $H(0)$ corresponding to $\lambda_0(0)$
(so $P_0(0)\phi^0 = \phi^0$).
Then define $\phi(t) = U(t)\phi^0$.
Note that
\[\begin{aligned}
P_0(t)\phi(t)
	&= (U(t) P_0(0) U\inv(t))(U(t)\phi^0)\\
	&= U(t) P_0(0)\phi^0 \\
	&= U(t) \phi^0 = \phi(t).
\end{aligned}\]
So we see that $\phi(t)$ is an eigenvector corresponding to the lowest eigenvalue of $H(t)$, for all $t \in [0,1]$.
Since $U(t)$ is invertible, $\phi(t)$ is non-zero, so we can normalize it to obtain
$\phi^*(t) \equiv \frac{\phi(t)}{\norm{\phi(t)}}$.
We have that $\phi^*(t)$ is a well-defined, normalized ground state of $H(t)$ for each $t$, and $\phi^*$ is continuous.
\end{proof}
%%%%%%%%%%%%%%%%%%%%%%%%%%%%%%%%%%%%%%%%%%%%%%%%%%%%%%%%%%%%%%
%%%%%%%%%%%%%%%%%%%%%%%%%%%%%%%%%%%%%%%%%%%%%%%%%%%%%%%%%%%%%%
\section{Optimization methods}
\label{sec:optimization_details}
We describe in detail the optimization methods used in our experiments.

%\subsubsection*{SGD}
\textbf{SGD} (stochastic gradient descent) is an optimization method defined by the iteration
\[
\theta^{k+1} = \theta^k - \gamma^k g^k.
\]
Here, $\gamma^k$ is the step length, and $g^k$ is the current gradient estimate.
Gradients are obtained using the parameter shift rule \cite{schuld2019evaluating,crooks_2019}
(we use RY-type ans\"atze for $\psi$, for which the parameter shift rule is applicable).
Along with evaluation of the objective estimate at the current parameter point, this makes the cost per iteration of the SGD method $2 p + 1$ objective evaluations.
Step lengths $\gamma^k$ decrease to zero with an asymptotic behavior like $\sfrac{1}{k}$;
this is consistent with most convergence theories \cite{bertsekas_2000}.
We use a step length equal to $1$ for the first ten iterations, then decrease it by setting it to $\sfrac{1}{(k-10)}$, for $k > 10$.

The addition of a momentum term is a common variant of SGD.
The idea is that this can help escape local minima;
see \cite{diakonikolas2019generalized,loizou2020momentum} for various perspectives and references.
Again, our point here is to show improvement in nearly any optimizer;
we use a basic form of SGD for simplicity.

%\subsubsection*{SPSA}
\textbf{SPSA} (simultaneous perturbation stochastic approximation) is a gradient-free stochastic optimization method \cite{spall92}.
It has been used successfully with VQE in previous studies \cite{kandala_2017}.
Each iteration consists of two objective evaluations, essentially giving a derivative approximation in a particular coordinate direction.
The algorithm randomizes these search directions and carefully controls the sequence of steplengths that it takes.
Hyperparameter settings are the default in Qiskit, with the exception of the initial step size which is set to $1.0$.

%\subsubsection*{NFT}
\textbf{NFT} (Nakanishi-Fujii-Todo) is a sequential coordinate minimization method \cite{nakanishi2020sequential}.
When $\psi$ is an RY-type ansatz, the objective $f$ is sinusoidal as a function of a single parameter $\theta_i$ with the others fixed.
The method leverages this fact to estimate $f$ as a function of a single parameter and analytically minimize it in this coordinate direction.
NFT requires two to three objective evaluations per iteration (the frequency of the extra evaluation is a hyperparameter -- we use the default in its implementation in Qiskit).
In the noise-free setting, at least, sequential minimization methods enjoy reasonable convergence guarantees \cite{wright2015coordinate}.
%%%%%%%%%%%%%%%%%%%%%%%%%%%%%%%%%%%%%%%%%%%%%%%%%%%%%%%%%%%%%%
%%%%%%%%%%%%%%%%%%%%%%%%%%%%%%%%%%%%%%%%%%%%%%%%%%%%%%%%%%%%%%
\section{Lithium hydride Hamiltonian}
\label{sec:hamiltonian}

Here we specify the Hamiltonian used in the lithium hydride ground state study of \Cref{sec:lih}.
We construct the Hamiltonian at an internuclear distance of $2.5${\AA} using PySCF \cite{pyscf} and the STO-3G basis set.
We consider the LiH molecule oriented with the bond along the x-direction.
For this orientation, the unoccupied 2p$_y$ and 2p$_z$ orbitals are perpendicular to the LiH bond and along with the core orbitals, are assumed frozen;
thus we do not consider the frozen orbitals in the quantum simulations.
This Hamiltonian is then transformed into the particle/hole representation \cite{barkoutsos_2018} and mapped to a qubit Hamiltonian using parity mapping \cite{bravyi_2002, seeley_2012}, which results in a four-qubit representation after applying qubit tapering to account for symmetries due to spin up and spin down particle number conservation \cite{bravyi2017tapering}.
The full Hamiltonian, written as a linear combination of tensor products of Pauli operators, has 100 terms.
See \Cref{tab:hamiltonian}.

\begin{table}
\caption{LiH Hamiltonian terms, to six decimal places.}
\label{tab:hamiltonian}
\centering
\begin{tabular}{|r|r|r|r|r|}\hline
0.567662 IIII &-0.025425 IZXX &0.013812 ZXXX &-0.011521 XZYY &-0.008083 ZXZX \\
0.245088 IIZI &0.025425 IZYY &0.013812 IXXX &0.011521 XIYY &-0.008083 IXZX \\
-0.245088 ZIII &0.025425 XXIZ &-0.013812 ZXYY &0.011521 XZXX &0.008083 ZXIX \\
-0.190085 IIZZ &-0.025425 YYIZ &-0.013812 IXYY &-0.011521 XIXX &0.008083 IXIX \\
-0.190085 ZZII &-0.019768 IIXZ &-0.013812 XXZX &0.010474 IIXX &-0.006835 ZXXZ \\
-0.107219 IZIZ &-0.019768 IIXI &0.013812 YYZX &-0.010474 IIYY &-0.006835 IXXZ \\
0.101581 IZII &0.019768 XZII &0.013812 XXIX &0.010474 XXII &-0.006835 ZXXI \\
-0.101581 IIIZ &-0.019768 XIII &-0.013812 YYIX &-0.010474 YYII &-0.006835 IXXI \\
0.098833 IZZI &-0.018582 XXZI &-0.012909 ZXZI &-0.009880 XZXI &-0.006835 XZZX \\
0.098833 ZIIZ &0.018582 YYZI &-0.012909 IXZI &0.009880 XIXI &0.006835 XZIX \\
-0.096556 ZIZI &0.018582 ZIXX &-0.012909 ZIZX &-0.009880 XZXZ &0.006835 XIZX \\
0.079438 ZZZZ &-0.018582 ZIYY &0.012909 ZIIX &0.009880 XIXZ &-0.006835 XIIX \\
-0.060240 ZZZI &0.017442 IZZX &-0.011861 XZZI &0.009298 ZZXI &-0.004511 ZXZZ \\
0.060240 ZIZZ &-0.017442 IZIX &0.011861 XIZI &0.009298 ZZXZ &-0.004511 IXZZ \\
-0.053253 IZZZ &0.017442 ZXIZ &-0.011861 ZIXZ &-0.009298 XZZZ &0.004511 ZZZX \\
0.053253 ZZIZ &0.017442 IXIZ &-0.011861 ZIXI &0.009298 XIZZ &-0.004511 ZZIX \\
0.033028 XXXX &0.016652 IZXZ &-0.011521 XXXZ &-0.009044 IIZX &-0.003631 XXZZ \\
-0.033028 YYXX &0.016652 IZXI &0.011521 YYXZ &0.009044 IIIX &0.003631 YYZZ \\
-0.033028 XXYY &0.016652 XZIZ &-0.011521 XXXI &0.009044 ZXII &0.003631 ZZYY \\
0.033028 YYYY &-0.016652 XIIZ &0.011521 YYXI &0.009044 IXII &-0.003631 ZZXX \\ \hline
\end{tabular}
\end{table}

%%%%%%%%%%%%%%%%%%%%%%%%%%%%%%%%%%%%%%%%%%%%%%%%%%%%%%%%%%%%%%
%%%%%%%%%%%%%%%%%%%%%%%%%%%%%%%%%%%%%%%%%%%%%%%%%%%%%%%%%%%%%%
\clearpage
\bibliography{bib/main}
%%%%%%%%%%%%%%%%%%%%%%%%%%%%%%%%%%%%%%%%%%%%%%%%%%%%%%%%%%%%%%
%%%%%%%%%%%%%%%%%%%%%%%%%%%%%%%%%%%%%%%%%%%%%%%%%%%%%%%%%%%%%%
\end{document}